\documentclass[conference]{IEEEtran}

\usepackage[utf8]{inputenc}
\usepackage{amsmath}
\usepackage{amsthm}
\usepackage{mathtools}
\usepackage{amsfonts}
\usepackage{hyperref}
\usepackage{url}
\usepackage{color}
\usepackage{tikz}
\usepackage{stackengine}
\usepackage{pgfplots}
\pgfplotsset{compat=1.13}
\usepackage{xargs}
\usepackage[noadjust]{cite}
\usepackage[compact]{titlesec}
\usepackage{caption}
\usepackage{subcaption}

\titlespacing*{\subsection}{0pt}{6pt}{3pt}
\newcommand\blankfootnote[1]{%
  \begin{NoHyper}%
  \let\svthefootnote\thefootnote%
  \let\thefootnote\relax\footnotetext{#1}%
  \let\thefootnote\svthefootnote%
  \end{NoHyper}%
}

\newcommand{\eps}{{\varepsilon}}        

\newcommand{\fr}{{\mathcal F}}

\newcommand{\NN}{{\mathbb{N}}}

\newcommand{\complex}{\mathbb{C}}
\newcommand{\reals}{\mathbb{R}}

\newcommand{\Bs}{{\mathcal{B}}}

\newcommand{\Ss}{{\mathcal{S}}}

\newcommand{\subgaussnorm}[1]{{\tau\left({#1}\right)}}

\newtheorem{theorem}{Theorem}

\newtheorem{cor}{Corollary}

\newtheorem{assume}{Assumption}
\newtheorem{remark}{Remark}
\newtheorem{definition}{Definition}

\newcommand{\numUsers}{{K}}
\newcommand{\indexUsers}{{k}}
\newcommand{\numChannelUses}{{M}}
\newcommand{\indexChannelUses}{{m}}
\newcommand{\channelVector}{{H}}
\newcommand{\noiseVector}{{N}}
\newcommand{\normalizedMessages}{{X}}
\newcommand{\messagesMatrix}{{Q}}

\newcommand{\rxVector}{{Y}}
\newcommand{\randomBaseVector}{{R}}
\newcommand{\transpose}[1]{{{#1}^T}}
\newcommand{\fadingTransformMatrix}{{A}}
\newcommand{\noiseTransformMatrix}{{B}}

\newcommand{\vectorNorm}[1]{\| {#1} \|_2}

\newcommand{\Expectation}{{\mathbb{E}}}
\newcommandx{\absolute}[3][1=\left, 2=\right]{#1 | #3 #2 |}
\newcommand{\tail}{{\varepsilon}}
\newcommand{\operatornorm}[1]{\| {#1} \|}
\newcommand{\frobeniusnorm}[1]{\| {#1} \|_F}
\newcommand{\Probability}{{\mathbb{P}}}
\newcommand{\powerconstraint}{{P}}

\newcommand{\identityMatrix}{{\mathbf{id}}}

\newcommand{\fadingApproxError}{{\eta}}
\newcommand{\uncorrelatedApprox}{A_i}
\newcommand{\errorProbOperatorNormTerm}{L}
\newcommand{\errorProbFrobeniusNormTerm}{F}
\newcommand{\errorProbDependenceTerm}{D}
\newcommand{\transmitPower}{{a}}

\begin{document}
\title{Over-The-Air Computation in Correlated Channels}

\author{
  \IEEEauthorblockN{
    Matthias Frey\IEEEauthorrefmark{1},
    Igor Bjelakovi\'c\IEEEauthorrefmark{2} %
    and %
    S\l awomir~Sta\'{n}czak\IEEEauthorrefmark{1}\IEEEauthorrefmark{2}\\ %
  } %
  \IEEEauthorblockA{
    \IEEEauthorrefmark{1}Technische Universität Berlin, Germany\\
    \IEEEauthorrefmark{2}Fraunhofer Heinrich Hertz Institute, Berlin, Germany 
  }%
  \vspace*{-2em}
}

\maketitle
\blankfootnote{
\vspace{-14pt}
This work was supported by the German Research Foundation (DFG) within their priority programs SPP 1798 ``Compressed Sensing in Information Processing'' and SPP 1914 "Cyber-Physical Networking", as well as under grant STA 864/7.

This work was partly supported by a Nokia University Donation. We gratefully acknowledge the support of NVIDIA Corporation with the donation of the DGX-1 used for this research.
}
\begin{abstract}
This paper addresses the problem of Over-The-Air (OTA) computation in wireless networks which has the potential to realize huge efficiency gains for instance in training of distributed ML models. We provide non-asymptotic, theoretical guarantees for OTA computation in fast-fading wireless channels where the fading and noise may be correlated. The distributions of fading and noise are not restricted to Gaussian distributions, but instead are assumed to follow a distribution in the more general sub-gaussian class. Furthermore, our result does not make any assumptions on the distribution of the sources and therefore, it can, e.g., be applied to arbitrarily correlated sources. We illustrate our analysis with numerical evaluations for OTA computation of two example functions in large wireless networks: the arithmetic mean and the Euclidean norm.
\end{abstract}

\section{Introduction and Prior Work}
\label{sec:intro}

Machine learning (ML) models are increasingly trained on data
collected by wireless sensor networks, with the goal to perform ML tasks in these
networks such as QoS prediction and anomaly detection. ML has undeniably a
great potential, but it is not available for free. The benefits of ML
must be set in relation to the effort and resources required. Since radio
communication resources (spectrum and energy) are generally scarce,
there is a strong interest~\cite{jordan2015machine} in resource-saving methods that
would allow ML models to be efficiently trained and used in resource-constrained wireless networks. In fact, great efforts have been made in recent years to reduce the communication overhead for training and deploying ML models~\cite{amiri2020machine,ahn2019wireless,mcmahan2017communication,konecny2016federated}.

However, for applications in massive networks of IoT devices, it is no longer sufficient to achieve constant or linear improvements, and instead the scaling in the number of users needs to be improved fundamentally,
otherwise the system performance may be severely degraded~\cite{gupta2000capacity}. Such
improvements can be achieved by abandoning the philosophy of
strictly separating the process of communication and
application-specific computation. For training of ML models and the computation of predictors, it is usually not necessary to accumulate the full data at a central point. Instead, it suffices to compute a function of the raw data (such as, e.g., the parameters of an ML model that fits the distributed training data). Since the amount of training data scales at least linearly with the number of users in the network while the complexity of the employed ML models could be constant or at least scale much more favorably, the paradigm shift from accumulating all of the data and processing it centrally towards over-the-air computation has the potential to realize fundamental gains in the scaling behavior of the communication cost.
Of course, these observations are known and have been used to improve
the performance of wireless systems in many scenarios~\cite{gastpar2003source,goldenbaum2013robust,goldenbaum2013harnessing,goldenbaum2014nomographic,nazer2007computation,zhan2009mimo,ordentlich2011practical,nazer2011compute,nazer2016expanding,goldenbaum2016harnessing}.

Against this background, in this paper, we advocate application-specific
schemes that take into account the underlying task, such as
computation of a function, directly at the physical layer.
The key ingredients that open up the door to such a paradigm shift in the
design of such schemes are provided by a fundamental result~\cite{kolmogorov1957representation,sprecher1965structure,buck1976approximate,buck1982nomographic} stating that every
function has a \emph{nomographic representation}. These representations allow us
to exploit the \emph{superposition property of the wireless channel}
for an efficient function computation over the ``air''. The
superposition (or broadcast) property is the ability of the wireless
channel to ``form'' (noisy) linear combinations of information-bearing
symbols transmitted by different nodes. This property is usually seen
as a problem in traditional wireless networks where it is the source
for interference between concurrent independent transmissions. If in
contrast, nodes cooperate for a common goal of function
computation, then the superposition property is in general not a
source for interference but rather for \emph{``inference''} which
should be exploited to compute functions. Whether the superposition
property can be exploited for performance gains strongly depends on
errors and uncertainties introduced by the wireless channel, and on how it combines the transmitted signals at the output. In
this paper,  we focus on the class $\fr_{\textrm{\textrm{mon}}}$ consisting of bounded and measurable nomographic functions, where the outer function is restricted in a way that allows for controlling the impact of noise and fading introduced by the channel. We derive bounds on the number of channel uses needed to approximate a function in $\fr_{\textrm{\textrm{mon}}}$ up to a desired accuracy and deal with a channel model that may exhibit a certain degree of correlation in noise and fading between users and over time.

Pre- and post-processing schemes for function approximation over fast-fading channels appeared in~\cite{kiril}, but no theoretical guarantees are derived. In~\cite{bjelakovic2019distributed}, we derived such theoretical guarantees for the case of independent fading and noise. \cite{liu2020over,dong2020blind} derived theoretical bounds on the error of over-the-air function computation, using either an assumption of known channel state information or of slow fading.

The direct application of over-the-air computation techniques to distributed gradient descent has received a lot of attention recently since such techniques can be used to solve the empirical risk minimization problem for ML models such as neural networks in the case of distributed training data without having to collect the training data at a central point~\cite{amiri2020machine,ahn2019wireless,mcmahan2017communication,konecny2016federated}.

\subsection{Paper Contributions}
Contrary to this prior work, in this paper we do not assume a particular source distribution on the transmitted messages; indeed, we derive uniform bounds that hold independently of how the sources are distributed and therefore also cover the important case of arbitrarily correlated sources. We focus on one-shot approximation of function values, which stands in contrast to a scenario in which the same functions are computed repeatedly, as in the case of the works that focus on the application to network coding. Our channel model is more general than in prior works, since we do not only consider sub-gaussian fading and noise, but generalize the system model from~\cite{bjelakovic2019distributed} to allow for correlations in fading and noise both between users and over time. Due to lack of space, we omit the details on how our results can be applied to ML problems. Two examples of such applications can be found in~\cite{frey2020over}.

%
%
\section{System Model and Problem Statement}     

%
%
\subsection{System Model}
\label{sec:systemmodel}
We consider the following channel model with $K$ transmitters and one receiver: For $m=1,\ldots, M$, the channel output at the $m$-th channel use is given by
\begin{equation*}
Y(m)=\sum\nolimits_{k=1}^{K} H_k(m)x_k(m)+ N(m).
\end{equation*}
$x_k(m)\in \complex$ for $k=1,\ldots, K$ and $m=1,\ldots, M$ are the transmit symbols which have to satisfy the peak power constraint $|x_k(m)  |^2 \le P$. $H_k(m)$, $k=1\ldots, K$, $m=1,\ldots, M$, are complex-valued random variables such that for every $m=1, \ldots ,M$ and $k=1,\ldots,K$, the real part $H_k^r (m)$ and the imaginary part $H_k^i (m)$  of $H_k(m)$ are zero-mean sub-gaussian\footnote{A random variable $X$ is called \emph{sub-gaussian} if $\subgaussnorm{X} :=  \inf \{t > 0: \forall \lambda \in \reals~~ \mathbb{E}\exp \left( \lambda (X - \mathbb{E} X) \right)   \le \exp \left( \lambda^2 t^2 / 2 \right)      \}$ is finite. $\subgaussnorm{\cdot}$ defines a semi-norm on the space of sub-gaussian variables and is called the \emph{sub-gaussian norm}. For more details see~\cite{frey2020over,buldygin,wainwright,vershynin}.} random variables with variance $1$. $N(m)$, $m=1,\ldots, M$, are complex-valued random variables, with both the real and the imaginary parts assumed to be centered and sub-gaussian.

\begin{definition}
\label{def:user-uncorrelated}
We say that the fading is \emph{user-uncorrelated} if for every $\indexUsers_1 \neq \indexUsers_2$, $j \in \{i,r\}$ and $\indexChannelUses$, the random variables $H^j_{\indexUsers_1}(\indexChannelUses)$ and $H^j_{\indexUsers_2}(\indexChannelUses)$ are independent.
\end{definition}

In order to introduce our dependency model for the fading and noise, we define 
\begin{align}
\label{eq:channelvector}
\channelVector := \transpose{(\channelVector(1), \dots, \channelVector(2\numChannelUses))}
\end{align}
where for $\indexChannelUses = 1, \dots, \numChannelUses$,
\begin{align*}
\channelVector(2\indexChannelUses-1) &:= (H_1^r(\indexChannelUses), \dots, H_K^r(\indexChannelUses))
\\
\channelVector(2\indexChannelUses) &:= (H_1^i(\indexChannelUses), \dots, H_K^i(\indexChannelUses)).
\end{align*}
So $\channelVector$ is the vector of all fading coefficients. Similarly, let 
\begin{align}
\label{eq:noisevector}
\noiseVector := \transpose{(N^r(1), N^i(1), \dots, N^r(\numChannelUses), N^i(\numChannelUses))}
\end{align}
be the vector of all the instances of additive noise.

\begin{assume}
\label{assume:correlation-model}
There exist a vector $\randomBaseVector$ of $(2\numUsers\numChannelUses + 2\numChannelUses)$ independent random variables with sub-gaussian norm at most $1$ and matrices $\fadingTransformMatrix \in \reals^{2\numUsers\numChannelUses \times (2\numUsers\numChannelUses + 2\numChannelUses)}$ and $\noiseTransformMatrix \in \reals^{2\numChannelUses \times (2\numUsers\numChannelUses + 2\numChannelUses)}$ such that $\channelVector = \fadingTransformMatrix \randomBaseVector$ and $\noiseVector = \noiseTransformMatrix \randomBaseVector$.
\end{assume}

In the case of $\randomBaseVector$ distributed i.i.d. standard Gaussian, Assumption~\ref{assume:correlation-model} amounts to a standard representation of arbitrarily correlated (and thus arbitrarily interdependent) multivariate Gaussian vectors. Therefore, by replacing $\randomBaseVector$ with a vector of independent sub-gaussian entries, we obtain a straightforward generalization of the Gaussian case which specializes to arbitrary correlations (which means that it also covers arbitrary stochastic dependence in the Gaussian case).

\begin{definition}
\label{def:user-uncorrelated-matrix}
We say that a matrix $\fadingTransformMatrix$ \emph{generates user-uncorrelated fading} if $H = \fadingTransformMatrix \randomBaseVector$ is user-uncorrelated according to Definition~\ref{def:user-uncorrelated}.
\end{definition}

The instantaneous realizations of the fading and noise need not be known at the transmitters, but statistical information on the average noise power and channel gain is necessary. This typically amounts to information about large scale fading at the transmitters only, but not about the small scale fading. While the error bounds depend on the structure of the correlation in fading and noise, the design of the pre- and post-processing operations does not require this knowledge.
%
%
\subsection{Distributed Approximation of Functions}
\begin{figure}
\begin{tikzpicture}
\coordinate                      (inK)          at (0,0);
\coordinate                      (in2)          at (0,2);
\coordinate                      (in1)          at (0,3);
\node[rectangle,draw]            (encK)         at (2,0) {$E_K^M$};
\node[rectangle,draw]            (enc2)         at (2,2) {$E_2^M$};
\node[rectangle,draw]            (enc1)         at (2,3) {$E_1^M$};
\node[rectangle,minimum height=3.5cm,align=center,draw] (channel) at (4,1.5) {$M$-fold\\channel};
\node[rectangle,draw]            (dec)          at (6,1.5) {$D_M$};
\coordinate                      (out)          at (8,1.5);
\node                            (vdots)        at (.8,1.2) {\Shortstack{. . . . . .}};

\draw[->] (inK) -- (encK) node[midway,above] {$s_K$};
\draw[->] (in2) -- (enc2) node[midway,above] {$s_2$};
\draw[->] (in1) -- (enc1) node[midway,above] {$s_1$};

\draw[->] (encK) -- (encK-|channel.west) node[midway,above] {$x_K^M$};
\draw[->] (enc2) -- (enc2-|channel.west) node[midway,above] {$x_2^M$};
\draw[->] (enc1) -- (enc1-|channel.west) node[midway,above] {$x_1^M$};

\draw[->] (channel) -- (dec) node[midway,above] {$Y^M$};
\draw[->] (dec) -- (out) node[midway,above] {$\tilde{f}$};
\end{tikzpicture}
\caption{System model. $E_1^M, \dots, E_K^M$ are randomized pre-processing functions, and $D^M$ is the post-processing function.}
\label{fig:system}
\end{figure}
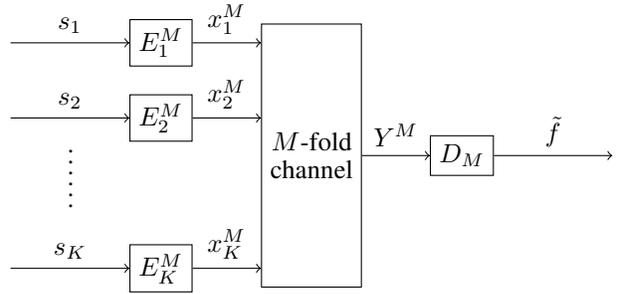

Our goal is to approximate functions  $f: \Ss_1\times \ldots \times \Ss_K\to\reals$  in a distributed setting. The sets $\Ss_1,\ldots \Ss_K\subseteq \reals $ are assumed to be closed and endowed with their natural Borel $\sigma$-algebras 
$\Bs(\Ss_1),\ldots ,\Bs(\Ss_K)$, and we consider the  product $\sigma$-algebra $\Bs (\Ss_1)\otimes \ldots \otimes \Bs(\Ss_K)$ on the set $ \Ss_1\times \ldots \times \Ss_K $. Furthermore, the functions  $f: \Ss_1\times \ldots \times \Ss_K\to\reals$ under consideration are assumed to be measurable.

A Distributed Function Approximation scheme (DFA) consists of pre- and post-processing functions as depicted in Fig.~\ref{fig:system}. So in order to approximate $f$, the transmitters apply their pre-processing maps to
$(s_1,\ldots, s_K)\in \Ss_1\times \ldots \times \Ss_K$
resulting in $E_1^M (s_1), \ldots, E_K^M(s_K)$ which are sent to the receiver using the channel $M$ times. 
The receiver observes the output of the channel and applies the recovery map $D^M$. The whole process defines an estimate $\tilde{f}$ of $f$.

Let $\varepsilon > 0,\delta \in (0,1)$ and $f: \Ss_1\times \ldots \times \Ss_K\to\reals $ be given. We say that $f$ is $\varepsilon$-approximated after $M$ channel uses with confidence level $\delta$ if there is a DFA
$(E^M,D^M)$ such that the resulting estimate $\tilde{f}$ of $f$ satisfies
\begin{equation}\label{eq:eps-delta-approx}
\mathbb{P}( |    \tilde{f} (s^K)- f(s^K)       |\ge \eps      )\le \delta
\end{equation}
for all $s^K:= (s_1, \ldots , s_K) \in \Ss_1\times \ldots \times \Ss_K $.
Let $M(f, \varepsilon, \delta)$ denote the smallest number of channel uses such that there is an approximation scheme $(E^M, D^M)$ for $f$ satisfying (\ref{eq:eps-delta-approx}). We call $M(f, \varepsilon, \delta)$ the communication cost for approximating a function $f$
with accuracy $\varepsilon$ and confidence $\delta$.

%
%
%
\subsection{The class of functions to be approximated}
%
%
%
\begin{definition}
\label{def:Fmon}
A measurable function $f: \Ss_1\times \ldots \times \Ss_K\to \reals$ is said to belong to $\fr_{\textrm{\textrm{mon}}}$ if there exist bounded and measurable functions
$(f_k)_{k \in \{1, \ldots , K\}}$, a measurable set $D\subseteq \reals$ with the property $f_1(\Ss_1)+\ldots + f_K(\Ss_K)\subseteq D$, a measurable function $F:D\to \reals$ such that
for all $(s_1, \ldots , s_K)\in \Ss_1\times \ldots \times \Ss_K$ we have
\begin{equation*}
f(s_1, \ldots, s_K)= F\left(  \sum\nolimits_{k=1}^K f_k(s_k)     \right ),
\end{equation*}
and there is a strictly increasing function $\Phi : [0, \infty) \to [0, \infty)$  with $\Phi(0)=0$ and
\begin{equation*}
| F(x)-F(y)|\le \Phi( | x-y | )
\end{equation*}
for all $x,y \in D$. We call $\Phi$ an \emph{increment majorant} of $f$.
\end{definition}
Notable examples of functions in $ \fr_{\textrm{\textrm{mon}}}$ include sums and weighted means as well as $p$-norms for $p \geq 1$. For more details and more general characterizations, see~\cite{bjelakovic2019distributed,frey2020over}.
%
%
We are now in a position to state our main theorem on approximation of functions in $ \fr_{\textrm{\textrm{mon}}}$.
To this end, we introduce the notion of total spread of  the inner part of $ f\in  \fr_{\textrm{\textrm{mon}}}$ as
\begin{equation*}
\bar{\Delta}(f):=\sum\nolimits_{k=1}^K ( \phi_{\max,k}-\phi_{\min,k}),
\end{equation*}
along with the $\max $-spread
\begin{equation}\label{eq:max-inner-spread}
\Delta (f):= \max_{1\le k \le K} ( \phi_{\max,k}-\phi_{\min,k}),
\end{equation}
where
 \begin{equation}\label{eq:phi-def-spread}
 \phi_{\min,k}:= \inf_{s \in \mathcal{S}_k} f_k (s), \quad \phi_{\max,k}:=\sup_{s \in \mathcal{S}_k} f_k( s).
 \end{equation}
 We define the relative spread with power constraint $P$ as
 \begin{equation*}
 \Delta (f\| P):= P \cdot \frac{\bar{\Delta} (f)}{\Delta (f)}.
 \end{equation*}
We use $\operatornorm{\cdot}$ and $\frobeniusnorm{\cdot}$ to denote the operator and Frobenius norm of matrices, respectively.
\begin{theorem}
\label{th:correlated}
Let $f\in  \fr_{\textrm{\textrm{mon}}}$, $M\in \NN$, and the power constraint $P \in \reals_{+} $ be given. Let $\Phi$ be an increment majorant of $f$. Assume the correlation in the fading $H$ and noise $N$ is described by $\fadingTransformMatrix$ and $\noiseTransformMatrix$ as described in Assumption~\ref{assume:correlation-model}. Let $\uncorrelatedApprox \in \reals^{2\numChannelUses\numUsers \times (2\numChannelUses\numUsers + 2\numChannelUses)}$ be a matrix which generates user-uncorrelated fading (as described in Definition~\ref{def:user-uncorrelated-matrix}), and let
\[
\fadingApproxError
:=
\operatornorm{(\fadingTransformMatrix+\uncorrelatedApprox)\transpose{(\fadingTransformMatrix-\uncorrelatedApprox)}}.
\]
Then there exist pre- and post-processing operations such that
\begin{multline}
\label{eq:correlated-statement}
\Probability\left(
  \absolute{\bar{f}-f(s_1, \dots, s_\numUsers)}
  \geq
  \tail
\right)
\\
\leq
\begin{aligned}[t]
&2\exp\left(
-
\frac{\numChannelUses \Phi^{-1}(\tail)^2}{
  16\errorProbFrobeniusNormTerm
  +
  \errorProbDependenceTerm
  +
  4\Phi^{-1}(\tail)
  \errorProbOperatorNormTerm
}
\right)
\\
&+
2\exp\left(
  -
  \frac{\numChannelUses \Phi^{-1}(\tail)^2}{
    256
    \errorProbFrobeniusNormTerm
    +
    32 \Phi^{-1}(\tail)
    \errorProbOperatorNormTerm
  }
\right),
\end{aligned}
\end{multline}
where
\begin{align*}
\errorProbOperatorNormTerm
&=
\left(
  \sqrt{\bar{\Delta}(f)} \operatornorm{\fadingTransformMatrix} + \sqrt{\frac{\Delta(f)}{\powerconstraint}}\operatornorm{\noiseTransformMatrix}
\right)^2
\\
\errorProbFrobeniusNormTerm
&=
\errorProbOperatorNormTerm
\left(
  \sqrt{\frac{\bar{\Delta}(f)}{\numChannelUses}} \frobeniusnorm{\fadingTransformMatrix} + \sqrt{\frac{\Delta(f)}{\powerconstraint\numChannelUses}} \frobeniusnorm{\noiseTransformMatrix}
\right)^2
\\
\errorProbDependenceTerm
&=
\bigg(
  4\sqrt{2\numChannelUses}
  \bar{\Delta}(f)
  \fadingApproxError
  +
  4
  \frac{\Delta(f)}{\sqrt{\powerconstraint\numChannelUses}}
  \frobeniusnorm{\fadingTransformMatrix\transpose{\noiseTransformMatrix}}
\bigg)^2.
\end{align*}
\end{theorem}
\begin{remark}
The matrix $\uncorrelatedApprox$ can be regarded as a user-uncorrelated approximation of $\fadingTransformMatrix$ and Theorem~\ref{th:correlated} leaves the choice of $\uncorrelatedApprox$ free. $\fadingApproxError$ is thus a measure for how strong the correlation is between different users. Important special cases are
\begin{itemize}
\item $\uncorrelatedApprox = \fadingTransformMatrix$ (i.e., the fading in the considered channel shows no correlations between the users) and consequently $\fadingApproxError = 0$
\item $\uncorrelatedApprox = 0$ (i.e., we have no user-uncorrelated approximation of $\fadingTransformMatrix$) and therefore 
$
 \fadingApproxError = \operatornorm{\fadingTransformMatrix}^2.
$
This shows us that $\uncorrelatedApprox$ in the sense of Theorem~\ref{th:correlated} always exists.
\end{itemize}
\end{remark}
\begin{cor}
(cf. \cite{frey2020over}) In the setting of Theorem~\ref{th:correlated} with uncorrelated fading and noise, i.e.,
\begin{align*}
\fadingTransformMatrix
:=
\begin{pmatrix}
\sigma_F \identityMatrix_{2\numChannelUses\numUsers} & 0
\end{pmatrix}
,~~
\noiseTransformMatrix
:=
\begin{pmatrix}
0 & \sigma_N \identityMatrix_{2\numChannelUses}
\end{pmatrix},
\end{align*}  
we have
\begin{multline*}
\Probability\left(
  \absolute{\bar{f}-f(s_1, \dots, s_\numUsers)}
  \geq
  \tail
\right)
\\
\leq
\begin{aligned}[t]
&2\exp\left(
-
\frac{\numChannelUses \Phi^{-1}(\tail)^2}{
  16\errorProbFrobeniusNormTerm'
  +
  4\Phi^{-1}(\tail)
  \errorProbOperatorNormTerm'
}
\right)
\\
&+
2\exp\left(
  -
  \frac{\numChannelUses \Phi^{-1}(\tail)^2}{
    256
    \errorProbFrobeniusNormTerm'
    +
    32 \Phi^{-1}(\tail)
    \errorProbOperatorNormTerm'
  }
\right),
\end{aligned}
\end{multline*}
where
\begin{align*}
\errorProbOperatorNormTerm'
&=
\left(
  \sqrt{\bar{\Delta}(f)} \sigma_F + \sqrt{\frac{\Delta(f)}{\powerconstraint}}\sigma_N
\right)^2
\\
\errorProbFrobeniusNormTerm'
&=
\errorProbOperatorNormTerm'
\left(
  \sqrt{2\numUsers\bar{\Delta}(f)} \sigma_F
  +
  \sqrt{\frac{2\Delta(f)}{\powerconstraint}}\sigma_N
\right)^2.
\end{align*}
\end{cor}

\begin{cor}
\label{cor:communication-cost}
As for the approximation communication cost, we have
\begin{align}\label{eq:commcost}
M(f,\varepsilon,\delta)
\leq
\frac{\log 4 - \log \delta}{\Phi^{-1}(\varepsilon)^2}
\Gamma,
\end{align}
where
\begin{multline*}
\Gamma
:=
\max\big(
  16\errorProbFrobeniusNormTerm
  +
  \errorProbDependenceTerm
  +
  4\Phi^{-1}(\tail)
  \errorProbOperatorNormTerm
  ,
  256
  \errorProbFrobeniusNormTerm
  +
  32 \Phi^{-1}(\tail)
  \errorProbOperatorNormTerm
\big).
\end{multline*}
\end{cor}
\begin{proof}
We upper bound (\ref{eq:correlated-statement}) as
\begin{align*}
 \mathbb{P} (  |    \bar{f} (s^K)- f(s^K)       |\ge \eps       )
 \le  
4\exp\left(
-\frac{\numChannelUses \Phi^{-1}(\tail)^2}{\Gamma}
\right),
\end{align*}
and solve the expression for $M$ concluding the proof.
\end{proof}
\begin{remark}
If $F$ is $C$-Lipschitz continuous, we can replace $\Phi^{-1}(\varepsilon)$ in (\ref{eq:commcost}) and the expression for $\Gamma$ with $\varepsilon/C$.
\end{remark}

\section{Sketch of the Proof of Theorem~\ref{th:correlated}}
In this section, we give the details of the pre- and post-processing steps that realize the error bounds given in the statement of Theorem~\ref{th:correlated}. The complete proof can be found in~\cite{frey2020over}. We only mention here that the proof requires a derivation of a slight variation of the Hanson-Wright inequality~\cite[Theorem 6.2.1]{vershynin}. The scheme presented here is a variation of the one that originally appeared in~\cite{kiril} and is based on concepts and ideas, e.g., from~\cite{gastpar2003source,goldenbaum2013robust,goldenbaum2013harnessing,goldenbaum2014nomographic}, but we extend the theoretical analysis to cover sub-gaussian correlated fading and noise and derive finite blocklength bounds.

\subsection{Pre-Processing}\label{sec:pre-proc}
In the pre-processing step we encode the function values $f_k(s_k)$, $k=1,\ldots ,K$ as transmit power:  
\begin{equation*}
X_k(m):= \sqrt{\transmitPower_k} U_k (m),
1\leq m\leq M
\end{equation*}
with $\transmitPower_k = g_k(f_k (s_k))$ and $g_k: [ \phi_{\min,k}, \phi_{\max,k}] \to [ 0, P  ] $,
\begin{equation*}
  g_k(t):= \frac{P}{\Delta(f)} (t-\phi_{\min,k}),
  \end{equation*}
 where $\Delta(f)$ is given in (\ref{eq:max-inner-spread}) and $\phi_{\min,k}$ is defined in (\ref{eq:phi-def-spread}).\\
 $U_k (m)$, $k=1, \ldots, K$, $m=1, \ldots ,M$ are i.i.d. with the uniform distribution on $\{-1,+1\}$. We assume the random variables $U_k (m)$, $k=1,\ldots,K$, $m=1\ldots,M$ are independent of
 $H_k(m)$, $k=1,\ldots, K$, $m=1,\ldots,M$, and $N(m)$, $m=1,\ldots,M$.
We write the vector of transmitted signals at channel use $\indexChannelUses$ as
$
\normalizedMessages(\indexChannelUses) := \left(X_1(m),\dots,X_K(m)\right)
$
and combine them in a matrix as
\begin{align*}
\begin{multlined}
\messagesMatrix
:=\\
\begin{pmatrix}
\normalizedMessages(1) & 0                      & 0                      & 0                      & 0      & \dots                                  & 0                                      \\
0                      & \normalizedMessages(1) & 0                      & 0                      & 0      & \dots                                  & 0                                      \\
0                      & 0                      & \normalizedMessages(2) & 0                      & 0      & \dots                                  & 0                                      \\
0                      & 0                      & 0                      & \normalizedMessages(2) & 0      & \dots                                  & 0                                      \\
                       &                        &                        &                        & \ddots &                                        &                                        \\
0                      & 0                      & 0                      & \dots                  & 0      & \normalizedMessages(\numChannelUses) & 0                                      \\
0                      & 0                      & 0                      & \dots                  & 0      & 0                                      & \normalizedMessages(\numChannelUses)
\end{pmatrix}.
\end{multlined}
\end{align*}
%
%
\subsection{Post-Processing}\label{sec:post-proc}
The vector $\rxVector$ of received signals across the $\numChannelUses$ channel uses can be written as
$
\rxVector = \messagesMatrix \cdot \channelVector + \noiseVector,
$
where $\channelVector$ and $\noiseVector$ are given in (\ref{eq:channelvector}) and (\ref{eq:noisevector}).
The post-processing is based on the received energy which has the form
\begin{align}
\label{eq:hwExpression}
\vectorNorm{\rxVector}^2
&=
\transpose{\rxVector} \rxVector
=
\transpose{(\messagesMatrix \fadingTransformMatrix \randomBaseVector + \noiseTransformMatrix \randomBaseVector)}
(\messagesMatrix \fadingTransformMatrix \randomBaseVector + \noiseTransformMatrix \randomBaseVector).
\end{align}

The receiver computes its estimate $\bar{f}$ of $f(s_1, \dots, s_K)$ as
\begin{align*}
\bar{f}
:=
F(\bar{g}(
  \vectorNorm{\rxVector}^2
  -
  \Expectation \vectorNorm{\noiseVector}^2
)),
\end{align*}
where
\[
\bar{g}(t) := \frac{\Delta(f)}{2 \cdot M \cdot  P}t + \sum\nolimits_{k=1}^K\phi_{\min,k}.
\]

\section{Numerical Results}
\begin{figure}
\centering
\vspace{-6pt}
\begin{tikzpicture} 
\begin{axis}[
  hide axis,
  xmin=10,
  xmax=50,
  ymin=0,
  ymax=0.4,
  legend columns = 4
]
\addlegendimage{mark=triangle};
\addlegendentry{~~mean~~~~};
\addlegendimage{mark=square};
\addlegendentry{~~norm~~~~};
\addlegendimage{};
\addlegendentry{~~DFA~~~~};
\addlegendimage{dashed};
\addlegendentry{~~TDMA};
\end{axis}
\end{tikzpicture}
\begin{subfigure}[b]{\textwidth}
\vspace{6pt}
\begin{tikzpicture}[scale=.97]
\definecolor{col1}{named}{red}
\definecolor{col2}{named}{blue}
\begin{semilogyaxis}[
  xlabel={noise power in dB},
  ylabel={mean squared error},
  ymin=1e-6,
  ymax=1e2,
  xmin=-20,
  xmax=50,
  legend style={legend pos=north west},
  legend columns=2
]
\addlegendimage{empty legend};
\addlegendentry{$K$};
\addlegendimage{empty legend};
\addlegendentry{$M$};
\addlegendimage{col1};
\addlegendentry{$40$};
\addlegendimage{empty legend};
\addlegendentry{$10^3$};
\addlegendimage{col2};
\addlegendentry{$2560$};
\addlegendimage{empty legend};
\addlegendentry{$10^5$};
\addplot[col1,mark=triangle] table [x=noise_power_dB, y=dist_users_40_chuses_1000, col sep=comma] {noniid_mean_noise.csv};
\addplot[col2,mark=triangle] table [x=noise_power_dB, y=dist_users_2560_chuses_100000, col sep=comma] {noniid_mean_noise.csv};
\addplot[col1,dashed,mark=triangle,mark options={solid}] table [x=noise_power_dB, y=tdma_users_40_chuses_1000, col sep=comma] {noniid_mean_noise.csv};
\addplot[col2,dashed,mark=triangle,mark options={solid}] table [x=noise_power_dB, y=tdma_users_2560_chuses_100000, col sep=comma] {noniid_mean_noise.csv};
\addplot[col1,mark=square] table [x=noise_power_dB, y=dist_users_40_chuses_1000, col sep=comma] {noniid_norm_noise.csv};
\addplot[col2,mark=square] table [x=noise_power_dB, y=dist_users_2560_chuses_100000, col sep=comma] {noniid_norm_noise.csv};
\addplot[col1,dashed,mark=square,mark options={solid}] table [x=noise_power_dB, y=tdma_users_40_chuses_1000, col sep=comma] {noniid_norm_noise.csv};
\addplot[col2,dashed,mark=square,mark options={solid}] table [x=noise_power_dB, y=tdma_users_2560_chuses_100000, col sep=comma] {noniid_norm_noise.csv};
\end{semilogyaxis}
\end{tikzpicture}
\subcaption{Mean squared error of the approximation schemes dependent\newline on the channel noise power.}
\label{fig:noise}
\end{subfigure}

\begin{subfigure}[b]{\textwidth}
\vspace{6pt}
\begin{tikzpicture}[scale=.97]
\definecolor{col1}{named}{red}
\definecolor{col2}{named}{blue}
\begin{loglogaxis}[
  xlabel={users},
  ylabel={mean squared error},
  ymin=1e-6,
  ymax=1e2,
  xmin=10,
  xmax=2560,
  legend style={at={(.5,.97)}, anchor=north},
  legend columns=2
]
\addlegendimage{empty legend};
\addlegendentry{noise};
\addlegendimage{empty legend};
\addlegendentry{$M$};
\addlegendimage{col1};
\addlegendentry{$-20$ dB};
\addlegendimage{empty legend};
\addlegendentry{$10^3$};
\addlegendimage{col2};
\addlegendentry{$20$ dB};
\addlegendimage{empty legend};
\addlegendentry{$10^5$};
\addplot[col1,mark=triangle] table [x=users, y=dist_noise_-20.0_chuses_1000, col sep=comma] {noniid_mean_users.csv};
\addplot[col2,mark=triangle] table [x=users, y=dist_noise_20.0_chuses_100000, col sep=comma] {noniid_mean_users.csv};
\addplot[col1,dashed,mark=triangle,mark options={solid}] table [x=users, y=tdma_noise_-20.0_chuses_1000, col sep=comma] {noniid_mean_users.csv};
\addplot[col2,dashed,mark=triangle,mark options={solid}] table [x=users, y=tdma_noise_20.0_chuses_100000, col sep=comma] {noniid_mean_users.csv};
\addplot[col1,mark=square] table [x=users, y=dist_noise_-20.0_chuses_1000, col sep=comma] {noniid_norm_users.csv};
\addplot[col2,mark=square] table [x=users, y=dist_noise_20.0_chuses_100000, col sep=comma] {noniid_norm_users.csv};
\addplot[col1,dashed,mark=square,mark options={solid}] table [x=users, y=tdma_noise_-20.0_chuses_1000, col sep=comma] {noniid_norm_users.csv};
\addplot[col2,dashed,mark=square,mark options={solid}] table [x=users, y=tdma_noise_20.0_chuses_100000, col sep=comma] {noniid_norm_users.csv};
\end{loglogaxis}
\end{tikzpicture}
\subcaption{Mean squared error of the approximation schemes dependent\newline on the number of participating transmitters.}
\label{fig:users}
\end{subfigure}

\begin{subfigure}[b]{\textwidth}
\vspace{6pt}
\begin{tikzpicture}[scale=.97]
\definecolor{col1}{named}{red}
\definecolor{col2}{named}{blue}
\begin{semilogyaxis}[
  xlabel={channel uses},
  ylabel={mean squared error},
  ymin=1e-6,
  ymax=1e2,
  xmin=1000,
  xmax=100000,
  legend style={at={(.97,.3)}, anchor=south east},
  legend columns=2
]
\addlegendimage{empty legend};
\addlegendentry{noise/dB};
\addlegendimage{empty legend};
\addlegendentry{$K$};
\addlegendimage{col1};
\addlegendentry{$-20$};
\addlegendimage{empty legend};
\addlegendentry{$40$};
\addlegendimage{col2};
\addlegendentry{$30$};
\addlegendimage{empty legend};
\addlegendentry{$2560$};
\addplot[col1,mark=triangle] table [x=chuses, y=dist_users_40_noise_-20.0, col sep=comma] {noniid_mean_chuses.csv};
\addplot[col2,mark=triangle] table [x=chuses, y=dist_users_2560_noise_30.0, col sep=comma] {noniid_mean_chuses.csv};
\addplot[col1,dashed,mark=triangle,mark options={solid}] table [x=chuses, y=tdma_users_40_noise_-20.0, col sep=comma] {noniid_mean_chuses.csv};
\addplot[col2,dashed,mark=triangle,mark options={solid}] table [x=chuses, y=tdma_users_2560_noise_30.0, col sep=comma] {noniid_mean_chuses.csv};
\addplot[col1,mark=square] table [x=chuses, y=dist_users_40_noise_-20.0, col sep=comma] {noniid_norm_chuses.csv};
\addplot[col2,mark=square] table [x=chuses, y=dist_users_2560_noise_30.0, col sep=comma] {noniid_norm_chuses.csv};
\addplot[col1,dashed,mark=square,mark options={solid}] table [x=chuses, y=tdma_users_40_noise_-20.0, col sep=comma] {noniid_norm_chuses.csv};
\addplot[col2,dashed,mark=square,mark options={solid}] table [x=chuses, y=tdma_users_2560_noise_30.0, col sep=comma] {noniid_norm_chuses.csv};
\end{semilogyaxis}
\end{tikzpicture}
\subcaption{Mean squared error of the approximation schemes dependent\newline on the number of channel uses.}
\label{fig:chuses}
\end{subfigure}
\caption{Visualization of the simulation results.}
\end{figure}
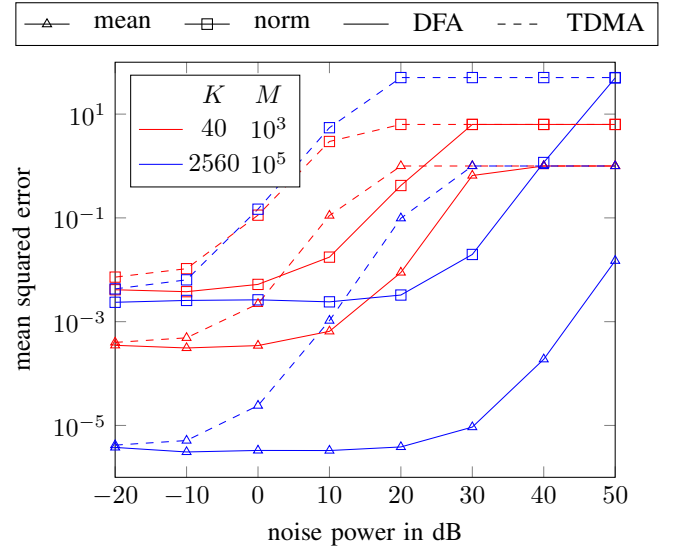
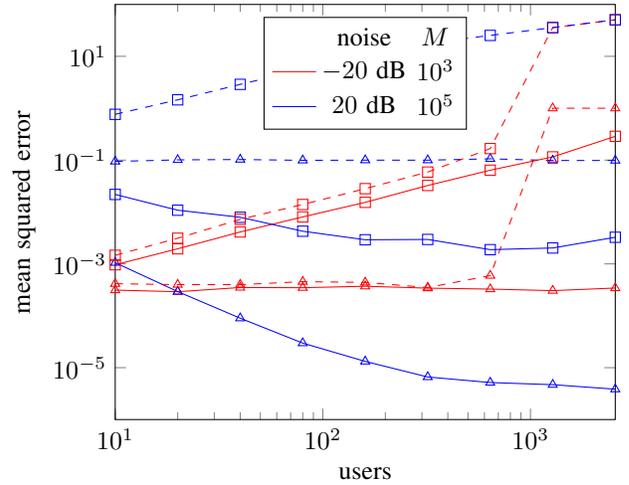
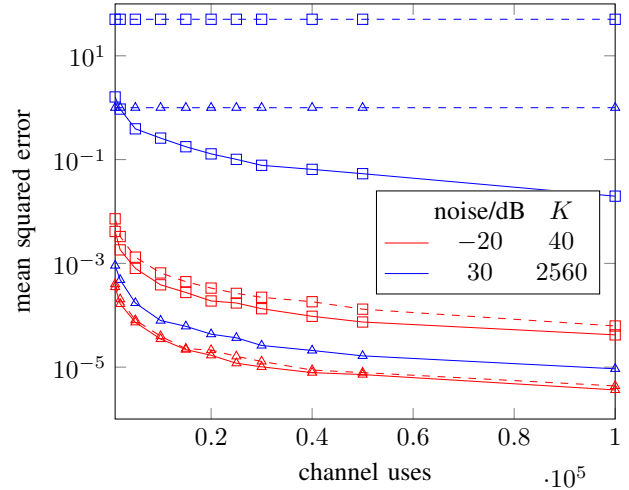

We have simulated the Distributed Function Approximation (DFA) scheme for Rayleigh fading channels with varying noise power, number of users and amount of channel resources. The simulations were done for two different functions, with the function arguments in both cases confined to the unit interval $[0,1]$, to highlight different aspects and properties of the scheme: The arithmetic mean function is linear and maps only to the interval $[0,1]$ (which means that no scheme can have an error larger than $1$), while the Euclidean norm function maps to $[0,\sqrt{K}]$ and can show how the DFA scheme deals with nonlinearities.

We compare with a simple TDMA scheme, in which each user transmits separately in its designated slot, protecting the analog transmission against channel noise in the same fashion as the DFA scheme, but not sharing the channel use with other transmitters. In the case where the number of channel uses available is much larger than the number of users sharing the resources, this form of a TDMA scheme is of course highly suboptimal, as the transmitters could use source and channel coding to achieve a higher reliability. However, such an approach is infeasible if the number of users is so high in comparison to the number of channel uses that only a few or possibly even less than one channel use is available to each user, and in this work we are mainly interested in the scaling behavior of our schemes in the number of users $K$. Therefore, this comparison provides an insight into the gain achieved by exploiting the superposition properties of the wireless channel while keeping in mind that for the regime of low $K$, there are better coded schemes available. We also remark that the DFA scheme only needs coordination between the transmitters insofar as all users need to transmit roughly at the same time, while a TDMA scheme necessitates an allocation of the channel uses to the individual transmitters, which can be costly in the case of high $K$. The simulations carried out in this section do not consider this scheduling problem and assume for the TDMA scheme that the time slots have already been allocated, and this knowledge is available at both the transmitters and the receiver. If $M < K$, there is not at least one channel use available to each user and the TDMA scheme can therefore not be carried out. We set the error in such cases to the maximum of $1$ or $\sqrt{K}$, respectively.

For the simulations, we assume a normalized peak transmitter power constraint of $1$ and channels with fading normalized to a variance of $1$ per complex dimension. The power of the additive noise is given in dB per complex dimension and its negative can therefore be considered as the signal-to-noise ratio (SNR). Each plotted data point is based on an average of $1000$ simulation runs.

The messages transmitted by the users are generated in the following way: First, we draw a value $\mu$, which is common to all transmitters, uniformly at random from $[0,1]$. We then draw the messages of all the users from a convex combination of the uniform distributions on $[0,\mu]$ and $[\mu,1]$ where we choose the weights in such a way that each message has expectation $\mu$. The reason for choosing this procedure although the DFA scheme also performs well for more natural distributions such as i.i.d. uniform in $[0,1]$ for all users is that in case of messages distributed according to a known i.i.d. distribution, the problem is too easy in the sense that both the mean and the Euclidean norm concentrate around values that depend only on the distribution and $K$, and therefore even without any communication at all, the function value can be quite accurately guessed if $K$ is large. On the other hand, we intend the DFA scheme for applications in which the messages can be correlated and distributed according to unknown distributions, so we opt for this form of correlation between the messages for the sake of the numerical evaluation.

In Fig.~\ref{fig:noise}, we can see that the DFA scheme is at least as good as the TDMA schemes for all the plotted data points and outperforms it in most cases, achieving a gain of up to $30$ dB for $K=2560$. For low powers of the additive noise, the effect of the multiplicative fading dominates, and therefore, the error saturates as the additive noise grows weaker. Fig.~\ref{fig:users} illustrates that the DFA scheme performs significantly better if the number of users is not too low, which is due to the superposition of the signals in the wireless channel resulting in a combined signal strength that grows with the number of users. We can also see the TDMA scheme performing similarly to the DFA scheme for low numbers of users, while quickly deteriorating in performance or even becoming infeasible as their number grows. In Fig.~\ref{fig:chuses}, we can observe the exponential decay of the error as the amount of channel resources used increases. Once again, we can observe that the TDMA scheme performs similarly to DFA for a low number of users, but becomes infeasible for larger $K$.

\bibliographystyle{IEEEtran}
\bibliography{references}
\end{document}